\documentclass[final]{elsarticle}
\usepackage{amsfonts,natbib}
\usepackage{amssymb,amsmath,latexsym}
\usepackage{amsthm,pslatex}
\usepackage{graphicx}

\newcommand{\rank}{{\rm rank}}

\newcommand{\gf}{{\rm GF}}

\newcommand{\calC}{{\mathcal C}}
\newcommand{\C}{{\mathcal C}}

\newcommand{\image}{{\mathrm{Im}}}

\newtheorem{theorem}{Theorem}[section]

\begin{document}

\begin{frontmatter}

\title{Another $q$-Polynomial Approach to Cyclic Codes}

\author{Can Xiang}
\address{College of Mathematics and Informatics, South China Agricultural University,
              Guangzhou, 510642, China. Email: cxiangcxiang@hotmail.com.}

\begin{abstract}
Recently, a $q$-polynomial approach to the construction and analysis of cyclic codes over $\gf(q)$
was given by Ding and Ling. The objective of this paper is to give another $q$-polynomial approach
to all cyclic codes over $\gf(q)$.
\end{abstract}


\end{frontmatter}

\section{Introduction}

Let $q$ be a power of a prime.
An $[n,k, d; q]$ linear code is a $k$-dimensional subspace of $\gf(q)^n$
with minimum nonzero (Hamming) weight $d$.

An $[n,k]$ linear code $\C$ over $\gf(q)$ is called {\em cyclic} if
$(c_0,c_1, \cdots, c_{n-1}) \in \C$ implies $(c_{n-1}, c_0, c_1, \cdots, c_{n-2})
\in \C$.
By identifying a vector $(c_0,c_1, \cdots, c_{n-1}) \in \gf(q)^n$
with the polynomial
$$
c_0+c_1x+c_2x^2+ \cdots + c_{n-1}x^{n-1} \in \gf(q)[x]/(x^n-1),
$$
any code $\C$ of length $n$ over $\gf(q)$ corresponds to a subset of the quotient ring
$\gf(q)[x]/(x^n-1)$.
A linear code $\C$ is cyclic if and only if the corresponding subset in $\gf(q)[x]/(x^n-1)$
is an ideal of the ring $\gf(q)[x]/(x^n-1)$.

It is well known that every ideal of $\gf(q)[x]/(x^n-1)$ is principal. Let $\C=\langle g(x) \rangle$ be a
cyclic code, where $g(x)$ is monic and has the smallest degree among all the
generators of the ideal $\C$. Then $g(x)$ is unique and called the {\em generator polynomial,}
and $h(x)=(x^n-1)/g(x)$ is referred to as the {\em parity-check} polynomial of $\C$.

Cyclic codes are widely employed in consumer electronics, data transmission devices,
broadcast systems, and computer systems as they have efficient encoding and decoding
algorithms. Cyclic codes have been studied for decades and a lot of progress has been made (see, for example, \cite{DLLZ2016,DY2013,Ding13-1,Ding13-2, HPbook,MacSlo} and the references therein).
Three approaches are generally used in the design and analysis of cyclic codes,
and based on a generator matrix, a generator polynomial and a generating idempotent, respectively.
These approaches have their advantages and disadvantages in dealing with cyclic codes.
Recently, a $q$-polynomial approach to the construction and analysis of cyclic codes over $\gf(q)$
was given by Ding and Ling \cite{DingLing}. Further progress was made in \cite{LYX}. The objective
of this paper is to present another $q$-polynomial approach to all cyclic codes over $\gf(q)$.

The remainder of this paper is organized as follows. Section \ref{sec-notations} introduces the
$q$-polynomial approach to cyclic codes developed in \cite{DingLing} and some recent progress
made in \cite{LYX}. Section \ref{sec-qpolycode} presents another $q$-polynomial approach to cyclic
codes over $\gf(q)$. Section \ref{sec-summ} summarizes this paper.

\section{The Ding-Ling $q$-polynomial approach to cyclic codes}\label{sec-notations}

Throughout this paper, we adopt the following notations unless otherwise stated:
\begin{itemize}
\item $p$ is a prime.
\item $q$ is a positive power of $p$.
\item $n$ is a positive integer, and denotes the length of a cyclic code
      over $\gf(q)$.
\item $r=q^n$.
\end{itemize}

Let $\lambda$ be an element of $\mathrm{GF}(r)^{*}$. A $q$-polynomial code with check element $\lambda$ is defined
by
\begin{eqnarray}\label{eq1}
\C_{\lambda}=\Big\{(c_{0},c_{1},\cdots,c_{n-1})\in \mathrm{GF}(q)^{n}: C(\lambda)=0,
\text{ \ where } C(x)=\sum_{i=0}^{n-1}c_{i}x^{q^{i}}\Big\}.
\end{eqnarray}
Ding and Ling defined the code $\C_{\lambda}$ and proved the following result \cite{DingLing}.

\begin{theorem}\label{lem1}
Every cyclic code of length $n$ over $\mathrm{GF}(q)$ can be expressed as a code $\C_{\lambda}$
for some element $\lambda\in \mathrm{GF}(q^{n})$, and is thus a $q$-polynomial code.
\end{theorem}

The Normal Basis Theorem in finite fields tells us that $\gf(r)$ has a {\em normal basis}
$\{\alpha, \alpha^q, \cdots, \alpha^{q^{n-1}}\}$ over
$\gf(q)$, where $\alpha \in \gf(r)^*$. Such an $\alpha$ is called a
{\em normal element} of $\gf(r)$ over $\gf(q)$.
Then $\lambda$ has a unique expression of the form
\begin{eqnarray}\label{eqn-betabasis}
\lambda = \sum_{i=0}^{n-1} \lambda_i \alpha^{q^i},
\end{eqnarray}
where all $\lambda_i \in \gf(q)$.

Ding and Ling also proved the following result on the dimension of the code $\C_{\lambda}$ \cite{DingLing}.

\begin{theorem}\label{thm-dimension}
Let $\lambda$ be as in (\ref{eqn-betabasis}), and let $\lambda=\gamma^s$ for some $s \ge 0$,
where $\gamma$ is a generator of $\gf(r)^*$.
The dimension of the code  $\calC_\lambda$ is equal to $n-\rank(B_\lambda)$, where the matrix
$B_\lambda$ is defined by
\begin{eqnarray}\label{eqn-matrixB}
B_\lambda = \left[ \begin{array}{llllll}
                         \lambda_0  & \lambda_{n-1} & \lambda_{n-2} & \cdots & \lambda_2 & \lambda_1 \\
                         \lambda_1  & \lambda_{0} & \lambda_{n-1} & \cdots & \lambda_3 & \lambda_2 \\
                         \lambda_2  & \lambda_{1} & \lambda_{0} & \cdots & \lambda_4 & \lambda_3 \\
                         \vdots  & \vdots & \vdots & \cdots & \vdots & \vdots \\
                         \lambda_{n-1}  & \lambda_{n-2} & \lambda_{n-3} & \cdots & \lambda_1 & \lambda_0
                         \end{array}
                \right],
\end{eqnarray}
and these $\lambda_i$ are defined in (\ref{eqn-betabasis}).

Furthermore, the dimension of the code  $\calC_\lambda$ is no less than $n-\ell_s$,
where $\ell_s$ denotes the size of the $q$-cyclotomic coset modulo $r-1$ containing $s$.
\end{theorem}

Lv, Yan and Xiao made further progress on $q$-polynomial codes by proving the following \cite{LYX}.

\begin{theorem}\label{lem2}
A cyclic code of length $n$ over $\mathrm{GF}(q)$ with parity-check polynomial $\lambda(x)=\lambda_{0}+\lambda_{1}x+\cdots+\lambda_{h}x^{h}$ $(h\leq n-1)$ is a $q$-polynomial code with check element $\lambda=\lambda_{0}\alpha+\lambda_{1}\alpha^{q}+\cdots+\lambda_{h}\alpha^{q^{h}}$, where all $\lambda_i\in \gf(q)$, and $\alpha$ is a normal element of $\mathrm{GF}(r)$ over $\mathrm{GF}(q)$.
\end{theorem}

Theorem \ref{lem2} also proves that every cyclic code over $\gf(q)$ is a $q$-polynomial code $\C_\lambda$.
It is interesting to notice that the proof of Theorem \ref{lem2} given by Lv, Yan and Xiao is much easier than
that of Theorem \ref{lem1} given by Ding and Ling \cite{DingLing}. BCH codes are a special type of cyclic codes.
Lv, Yan and Xiao also defined the $q$-BCH codes, which can be viewed as an analogue of the classical BCH codes
\cite{LYX}.

In the next section, we will give another $q$-polynomial treatment of all cyclic codes over $\gf(q)$.

\section{Another $q$-polynomial approach to cyclic codes over $\gf(q)$}\label{sec-qpolycode}

A $q$-polynomial over $\gf(q)$ is a polynomial of the form
$\ell(x)=\sum_{i=0}^{h} \ell_i x^{q^i}$ with all coefficients
$\ell_i$ in $\gf(q)$ and $h$ being a nonnegative integer.

\subsection{Description of $q$-polynomial image codes}\label{sec-qpcode}

Let
\begin{eqnarray}\label{eqn-ellpolynom}
\ell_{(n, q)}(x)=\sum_{i=0}^{n-1} \ell_i x^{q^i}
\end{eqnarray}
be a $q$-polynomial over $\gf(q)$, which is a subfield of $\gf(q^n)$. We define
\begin{eqnarray}\label{eqn-image}
\image(\ell_{(n, q)})=\{\ell_{(n, q)}(y): y \in \gf(q^n)\}.
\end{eqnarray}

Given a normal basis $\{\beta, \beta^{q}, \beta^{q^2}, \cdots, \beta^{q^{n-1}}\}$ of $\gf(q^n)$
over $\gf(q)$, every $c \in \gf(q^n)$ can be expressed as
$$
c=\sum_{i=0}^{n-1} c_i \beta^{q^i}.
$$
This gives a one-to-one correspondence between $\gf(q^n)$ and $\gf(q)^n$ with the fixed normal element $\beta$.

Define
\begin{eqnarray}\label{eqn-mmcode}
\calC_{(n, \beta, \ell)}=\left\{ (c_0, c_1, \cdots, c_{n-1}) \in \gf(q)^n: c=\sum_{i=0}^{n-1} c_i \beta^{q^i}
\in \image(\ell_{(n, q)}) \right\}.
\end{eqnarray}
Clearly, $\calC_{(n, \beta, \ell)}$ is a linear code over $\gf(q)$ of length $n$, and is called the
{\em $q$-polynomial image code} of the  $q$-polynomial $\ell_{(n, q)}(x)$ of (\ref{eqn-ellpolynom}).

\subsection{Main results about the $q$-polynomial image code $\calC_{(n, \beta, \ell)}$}\label{sec-genematrix}

The main results of this paper are stated in the following two theorems.

\begin{theorem}\label{thm1}
Let symbols and notations be as above. Then the generator matrix of the $q$-polynomial image code $\calC_{(n, \beta, \ell)}$
is given by
\begin{equation}\label{genematrixG}
    G=
    \left[
      \begin{array}{ccccccccc}
        \ell_0 & \ell_{n-1} & \ell_{n-2} & \cdots & \ell_2 & \ell_1 \\
        \ell_1 & \ell_{0} & \ell_{n-1}   & \cdots & \ell_3 & \ell_2 \\
        \vdots & \vdots & \vdots & \cdots & \vdots & \vdots \\
        \ell_{n-1} & \ell_{n-2} & \ell_{n-3} & \cdots & \ell_1 & \ell_0 \\
      \end{array}
    \right],
\end{equation}
and thus $\calC_{(n, \beta, \ell)}$ is a cyclic code over $\gf(q)$ of length $n$.
\end{theorem}

\begin{proof}
By the Normal Basis Theorem, every $y \in \gf(q^n)$ has a unique expression of the form
$$
y=\sum_{i=0}^{n-1} y_i \beta^{q^i},
$$
where all $y_i\in \gf(q)$. It follows from
(\ref{eqn-ellpolynom}) that
\begin{eqnarray}\label{eqn-ellpolynom1}
\ell_{(n, q)}(y)&=&\sum_{i=0}^{n-1} \ell_i y^{q^i} \nonumber \\
&=&\sum_{i=0}^{n-1} \ell_i \left(\sum_{i=0}^{n-1} y_i \beta^{q^i}\right)^{q^i} \nonumber \\
&=&\sum_{i=0}^{n-1} \left(\sum_{j=0}^{n-1}\ell_j~y_{(i-j) \bmod{n}}\right)~\beta^{q^i},
\end{eqnarray}
where all $\ell_i\in \gf(q)$ and $y_i \in \gf(q) $.

By (\ref{eqn-image}) and (\ref{eqn-ellpolynom1}), it is easily seen that the code $\calC_{(n, \beta, \ell)}$ of (\ref{eqn-mmcode}) can be
expressed as
\begin{align*}
\calC_{(n, \beta, \ell)}=&\Big\{ (c_0, c_1, \cdots, c_{n-1}) \in \gf(q)^n: c_i=\sum_{j=0}^{n-1}\ell_j~y_{(i-j)~mod~n}~\\
&\text{ \ for all } y=(y_1,y_2,\cdots,y_n)\in \gf(q)^n  \Big\}.
\end{align*}
This leads to
\begin{equation}\label{eq7}
    \left[
      \begin{array}{cccccc}
        \ell_0 & \ell_{n-1} & \ell_{n-2} & \cdots & \ell_2 & \ell_1 \\
        \ell_1 & \ell_{0} & \ell_{n-1}   & \cdots & \ell_3 & \ell_2 \\
        \vdots & \vdots & \vdots & \cdots & \vdots & \vdots \\
        \ell_{n-1} & \ell_{n-2} & \ell_{n-3} & \cdots & \ell_1 & \ell_0 \\
      \end{array}
    \right]
    \left[
      \begin{array}{c}
        y_{0} \\
        y_{1} \\
        \vdots\\
        y_{n-1} \\
      \end{array}
    \right]
    =
    \left[
      \begin{array}{c}
        c_{0} \\
        c_{1} \\
        \vdots\\
        c_{n-1} \\
      \end{array}
    \right].
\end{equation}
Hence, $G$ is the generator matrix of the code $\calC_{(n, \beta, \ell)}$. Since $G$ is a circulant matrix,
$\calC_{(n, \beta, \ell)}$ is a cyclic code over $\gf(q)$ of length $n$ and its dimension is equal to
$\textmd{rank}(G)$. This completes the proof.
\end{proof}

\begin{theorem}\label{thm2}
Let symbols and notations be as above. Then every cyclic code over $\gf(q)$ of length $n$ can be expressed
as a code $\calC_{(n, \beta, \ell)}$ of (\ref{eqn-mmcode}) for a normal element $\beta \in \gf(q^n)$ and
a $q$-polynomial $\ell(x)$ over $\gf(q)$.
\end{theorem}

\begin{proof}
For any cyclic code $C$ over $\gf(q)$ of length $n$, we suppose that $C$ is an $[n,n-s]$ cyclic code over $\gf(q)$ with  generator polynomial $g(x)=\sum_{i=0}^{s}g_ix^i$, where all $g_i \in \gf(q)$. Then
\begin{equation}\label{genematrixG1}
    G_1=
    \left[
      \begin{array}{lllllllll}
        g_0 & g_1 & g_2 & \cdots & g_s & 0 & 0 & \cdots & 0 \\
        0 & g_0 & g_1 & g_2 & \cdots & g_s & 0 & \cdots & 0 \\
        0 & 0 & g_0 & g_1 & g_2 & \cdots & g_s  & \cdots & 0 \\
        \vdots & \vdots & \ddots & \ddots & \ddots & \ddots & \cdots & \ddots & \vdots \\
        0 & 0 & 0 & 0 & g_0 & g_1 & g_2 & \cdots & g_s  \\
        g_s & 0 & 0 & 0 & 0 & g_0 & g_1 & \cdots & g_{s-1}  \\
        g_{s-1} & g_s & 0 & 0 & 0 & 0 & g_0 & \cdots & g_{s-2}  \\
        \vdots & \vdots & \vdots & \vdots & \vdots & \vdots & \vdots & \ddots & \vdots  \\
        g_1 & g_2 & g_3 & \cdots & g_s & 0 & 0 & \cdots & g_0
      \end{array}
    \right]
\end{equation}
is a generator matrix of the code $C$ and $\textmd{rank}(G_1)=n-s$.

Let  $\beta$ be a normal element of $\gf(q^n)$ over $\gf(q)$ and $\ell(x)=\sum_{i=0}^{n-1} \ell_i x^{q^i} \in \gf(q)[x]$,
where
\begin{eqnarray*}
&& \ell_0 = g_0, \\
&& \ell_i =0 \mbox{ for all $i$ with } 1 \leq i \leq n-s-1, \\
&& \ell_{n-i} = g_i  \mbox{ for all $i$ with } 1 \leq i \leq s.
\end{eqnarray*}
It then follows from the proof of Theorem \ref{thm1} that the code $\calC_{(n, \beta, \ell)}$ has generator matrix
$G_1$. Hence, $C=\calC_{(n, \beta, \ell)}$. This completes the proof.
\end{proof}

\section {Concluding remarks}\label{sec-summ}

In this paper, we proposed another $q$-polynomial approach to all cyclic codes over $\gf(q)$, and gave a specific
expression of their generator matrices. It would be very interesting to employ this approach to find out cyclic codes with
desirable parameters.

\end{document}